%% file: main.tex
\documentclass[preprint]{elsarticle}

\usepackage{lineno,hyperref}
\usepackage{command}

\usepackage{amsfonts}
\usepackage{amsmath}
\DeclareMathOperator*{\argmax}{arg\,max}

\usepackage[svgnames]{xcolor}		%
\colorlet{MyRed}{Crimson!75!Black}
\colorlet{MyGreen}{DarkGreen!80!Black}
\colorlet{MyBlue}{MediumBlue}

\usepackage[showdeletions]{color-edits}		%
\usepackage[normalem]{ulem}		%

\addauthor{B}{red}

\addauthor{PM}{MyBlue}

\begin{document}

\begin{frontmatter}

\title{A heuristic for estimating Nash equilibria in first-price auctions with correlated values}

\author[criteo]{Benjamin Heymann}
\author[cnrs,criteo]{Panayotis Mertikopoulos}

\address[criteo]{Criteo AI Lab}
\address[cnrs]{Univ. Grenoble Alpes, CNRS, Inria, Grenoble INP, LIG, 38000, Grenoble, France}

\begin{abstract}
\input{abstract}
\end{abstract}

\end{frontmatter}

%\linenumbers

\section{Introduction}

\input{introduction}

\subsection*{Related work}
\input{Related_works}

\section{Formulation of the game, and agent-based reformulation}
\input{introduction-reformulation}
\subsection{The standard approach to model an auction as a game}

\input{standard_game_formulation}

\subsection{The agent-based auction model}
\label{subsec-std-formulation}

\input{agent_perspective}
\subsection{Equivalence of the two formulations}

\input{equivalence}

\section{Fictitious bidding: definition, implementation, and results}
\subsection{Definition of the heuristic}

\input{algorithm}

\subsection{Implementation and results}

\input{implementation}

\input{examples}

\section{Conclusion}
\input{discussion}

\footnotesize
\setlength{\bibsep}{\smallskipamount}
\bibliographystyle{elsarticle-num}
\bibliography{mybibfile,Bibliography-PM}

\end{document}

%% file: abstract.tex
Our paper concerns the computation of Nash equilibria of first-price auctions with correlated values.
While there exist several equilibrium computation methods for auctions with \textit{independent} values, the correlation of the bidders' values introduces significant complications that render existing methods unsatisfactory in practice.
Our contribution is a step towards filling this gap:
inspired by the seminal fictitious play process of Brown and Robinson, we present a learning heuristic -- that we call \textit{fictitious bidding} (FB) -- for estimating Bayes-Nash equilibria of first-price auctions with correlated values, and we assess the performance of this heuristic on several relevant examples.

%% file: introduction.tex
In single-item sealed bid first-price auctions (sometimes referred to as \textit{pay-as-bid}), several potential buyers bid simultaneously for an item to be bought;
the highest bidder then  gets the item and pays their bid to the seller.
A common approach to model a sealed bid first-price auction is to frame it as a game. 
More precisely, since the buyers’ valuations are not necessarily known by the other participants, 
one usually supposes that each of those valuations are random variables sampled from a probability distribution.  In this narrative,  each buyer observes their own valuation \textendash\ a \emph{private} signal \textendash\ before submitting their bid, so that the resulting situation belongs to the class of \textit{Bayesian} games.
A good share of the literature on auctions makes the additional structuring assumptions that the valuations are independently distributed.

Given the probability distributions of the  buyers' values, an important research question is to determine the outcome of the auction.
While several solution concepts exist in the game-theoretic literature, the notion of Bayes-Nash equilibrium is of primary theoretical and practical importance for the study of  auctions, and the numerical estimation of   Bayes-Nash equilibrium has been a vigorously researched question for quite some time, with several breakthroughs along the way.
However, a major challenge that arises in practice is that
(\textit{a}) the above methods invariably rely on a first-order characterization of the solution
and
(\textit{b}) they require the bidders' values to be independently distributed, which is a very strong limitation for real applications.

\subsection*{Our contributions}
Our paper seeks to overcome this limitation via a learning approach whereby the bidders' interactions are stylized as a repeated game, and the auction's Bayes-Nash equilibrium emerges as the players' empirical bid distribution profile over time.
This heuristic \textendash\ both simple and natural \textendash\ specifies the bidders' behavior and simulates the dynamics of the bidders' strategy profiles as the auction is repeated, and shares some features with the seminal fictitious play procedure introduced by Brown and Robinson as a solution method for zero-sum games. This is why we name it fictitious bidding (FB).

Given the immense difficulty of calculating Bayes-Nash equilibria in first-price auctions -- a problem which is PP-complete, and may even be PSPACE-complete \cite{CP14} -- any theoretical guarantees of our method would necessarily concern smaller, more tractable classes of first-price auctions.
To avoid narrowing down our focus in this way, we forego a theoretical analysis of the proposed heuristic, and we undertake instead a numerical evaluation campaign in a series of diverse, relevant examples.
Also, given that our heuristic does not rely on the structure of the first-price auction, it is also tested successfully on other auction rules.
In all the examples we tested, the method provides an equilibrium outcome in a few seconds or minutes, depending on the example, a feature which we consider particularly appealing for practical applications.

To the best of our knowledge, the proposed heuristic is the first approach in the literature that remains agnostic to any correlations or dependencies between the users' value distributions.
Moreover, the adopted learning viewpoint offers the distinct advantage of being oblivious to the actual auction mechanism (and any first-order characterization of the auction's equilibrium), so the proposed heuristic can also be applied directly to other auction settings -- such as mixtures of first and second-price auctions or other combinations thereof.

%% file: related_works.tex
The literature on auctions is immense~\cite{krishna2009auction}, and is impossible to survey in a short paper;
we therefore discuss below only the most relevant works we are aware of concerning the numerical computation of Bayes-Nash equilibria in first-price auctions.

In his seminal paper~\cite{vickrey1961counterspeculation}, Vickrey derived a characterization of Nash equilibria when the bidders' values are independent and the bidders otherwise identical.
Later, Plum~\cite{plum1992characterization} showed how to compute the Nash equilibrium of $2$-bidder auctions for some special cases when the price is a combination of first and second price.
A first general numerical method for computing the Nash equilibrium of a first-price auction with independent values appears in~\cite{marshall1994numerical}.
Theoretical analyses of the equilibrium structure are provided in particular in~\cite{10.2307/2648842,maskin2003uniqueness,reny2004existence}.
In order to study bidding rings, Bajari introduces several heuristics to compute the Nash equilibrium of first-price auctions~\cite{bajari2001comparing}.
Several other computation methods for first-price auctions with independent values have been proposed since then~\cite{Gayle2008,fibich2011numerical,kaplan2012asymmetric,Fibich2012,hubbard2014numerical}, and more recently,~\cite{10.5555/3398761.3398929}, that addresses the cases of discrete value distributions.
These methods rely on a first-order characterization of the best reply of the players to produce a system of ordinary differential equations that are then solved using various methods. One of the main difficulty lies in the numerical instability of the solution to the system.

Research on first-price auctions received renewed impetus in 2019 when Google switched its display advertising market place to first-price~\cite{heymann2020bid,paes2020competitive}.
This led to a surge of interest in new topics such as 
computational complexity~\cite{filos2021complexity},
numerical approximation~\cite{rasooly2021importance}
or
more recently, the use of neural networks to compute the auction's equilibrium \cite{bichler2021learning}.

We also mention that there is a very active track of research that, in the of wake of~\cite{myerson1981optimal}, aims at maximizing the seller's revenue with pragmatism~\cite{hartline2009simple,dhangwatnotai2015revenue,cole2014sample}.

Our paper takes a complementary, online learning approach to the above:
the bidders' interactions are modeled as a repeated game that unfolds in discrete time, with each agent playing an approximate "best response" to the observed empirical distribution of bids.
This heuristic closely resembles -- and was inspired by -- the fictitious play process as pioneered concurrently by Brown \cite{Bro51} and Robinson \cite{Rob51}.
This is one of the most widely studied procedures for learning in games, and it involves each player playing a best response to her believs about her opponents, given here by the the empirical frequency of past play.

The convergence of these beliefs to Nash equilibrium was first established in two-player zero-sum games by Robinson \cite{Rob51}.
Subsequently, the method has been shown to converge in
$2\times2$ games \citep{Miy61},
general $N$-player potential games \citep{MS96-jet},
symmetric games with an interior evolutionarily stable strategy \citep{Hof95b},
as well as certain classes of supermodular games \citep{MR90,MR91,Kri92,Hah08}.
Variants of fictitious play involving a certain degree of explicit exploration / randomness have also been considered in the literature:
the most widely studied of these processes is that of \textit{stochastic} -- or \textit{perturbed} -- fictitious play, which was introduced by Fudenberg \& Kreps \citep{FK93}, and which was shown by Hofbauer \& Sandholm \citep{HS09} to converge to an approximate Nash equilibrium -- a quantal response equilibrium to be exact -- in the same classes of games as fictitious play.%
\footnote{Stochastic fictitious play is related -- \textit{but not in any way equivalent} -- to the class of no-regret learning policies known as "follow the regularized leader" \citep{SS11};
for a detailed discussion, we refer the reader to \cite{MS16,HLMS21a}, and references therein.}

At the same time, the literature on the convergence of (stochastic) fictitious play should not be interpreted as suggestion that these processes converge to equilibrium in \textit{all} games:
notable examples include Jordan's three-player matching pennies variant \citep{Jor93}, as well as the counterexamples of Shapley \citep{Sha64} and Gaunersdorfer and Hofbauer \citep{GH95}.
In a similar vein, we should stress that we do not make any claims of global convergence to Bayes-Nash equilibrium in \textit{all} first-price auctions.
However, the series of numerical examples presented is sufficiently wide in scope and breadth to provide reasonable optimism for the use of the proposed heuristic in practice.

%% file: introduction-reformulation.tex
The choice of mathematical formalism to model the bidders' behavior (their \emph{strategies}) plays a crucial role in the analysis. 
For instance, when the support of the bidders' values is an interval, it is well known that, under mild assumptions, there exists a unique pure Nash equilibrium~\cite{lebrun1999first,athey2001single,maskin2003uniqueness}, so randomization is not needed for the analysis.
By contrast, if the bidders' values belong to a discrete set, bid randomization is required to ensure the existence of a Nash equilibrium~\cite{wang2020bayesian}.
In both cases --- discrete or continuous support ---  most of the literature captures the bidders' strategies via their cumulative bid distributions.
This is natural because, in both cases, the bid cannot decrease when the value increases.

However, this formalism is not well-adapted to our algorithmic proposal.
For this reason, after providing a quick refresher on the standard first-price auction model, we present an alternative, agent-based formulation which is much better suited to the discussion of our learning heuristic.

%% file: standard_game_formulation.tex
  A standard way to model one-item auctions is as follows:  The auction is framed as a Bayesian game with players (the auction participants) belonging to a finite set $\PLAYER$. 
    Each  player $\player\in\PLAYER$ has access to a (random) private signal $\val_\player\in \RR^+$ that corresponds to their value for the item being auctioned, and places a \textit{bid} $\bid_\player=\beta_\player(\val_\player)$.
    The auction is then cleared\footnote{In case of equality of the highest bids, a tie breaking rule is needed.} and 
     the payoof of the $\player$-th player in a first-price auction is given by  $(\val_\player - \bid_\player) \prod_{j\in\PLAYER;j\neq \player}\{\bid_\player>\bid_j\}$.
     
    Formally, a one-item auction is a tuple $(\PLAYER,(\valueSet_\player)_{\player\in\PLAYER},F)$ such that
    \begin{enumerate}
        \item $\PLAYER$ is a finite set;
        \item Each $\valueSet_\player$ is a compact subset of $\RR^+$;
        \item $F$ is a distribution on $\prod_{\player\in\PLAYER}\valueSet_\player$.
    \end{enumerate}
    For such an auction, a strategy  for player $\player\in\PLAYER$ is a measurable map $\beta_\player\colon\valueSet_\player\to\RR^+$.
    We denote by $\Sigma_\player$ the set of  strategies of player $\player$ so
    the expected payoff of player $\player$
    given a strategy profile $(\beta_\player)_{\player\in\PLAYER}$ is 
    \begin{equation}
    \pi_\player(\beta_\player,\beta_{-\player}) = \EE_F \big[ (v_\player - \beta_\player(\val_\player)) \prod_{j\neq \player}\{\beta_\player(\val_\player)>\beta_j(\val_j)\} \big],
    \end{equation}
    where we have used the standard game-theoretic shortand $-\player=\PLAYER\setminus\{\player\}$. A (pure strategy) Nash Equilibrium (NE) is then defined as a strategy profile $\beta_\PLAYER$ such that
 \begin{equation}
        \pi_\player(\beta_\player,\beta_{-\player}) \geq\pi_\player(\beta'_\player,\beta_{-\player}) \quad\forall \beta_\player\in\Sigma_\player, \quad \forall \player\in\PLAYER.
 \end{equation}
 In the remaining of this paper, we will refer to this formulation as the \emph{player-based} formulation.
 
 As we mentioned above, when the value sets $(\valueSet_\player)_{\player\in\PLAYER}$ are intervals, there exists under mild assumptions a (unique) Nash equilibrium in pure strategies, so there is no need to introduce mixed strategies. 
 However, such mixed strategies are required in a discrete setting where the existence of a Nash equilibrium is only guaranteed in \textit{mixed} strategies.
 We will therefore introduce mixed strategies only for the discrete setting because, as observed by Aumann~\cite{aummanmixed}, the definition of mixed strategies when pure strategies are continuous mappings between two continuous sets introduces significant technical difficulties that are beyond the scope of this work.%
 \footnote{The problem is that one cannot write in this case $\Delta\Sigma_\player$.}
 We hence call a mixed strategy a probability distribution on $\Sigma_\player$, and denote by  $\Delta\Sigma_\player$ the set of mixed strategies of player $\player$.

%% file: agent_perspective.tex
We propose a variant to  the  \textit{player-based}  auction model introduced in Section~\ref{subsec-std-formulation}.
An agent-based auction model is a tuple $(\AGENT,\val_\AGENT,\probascenario)$ such that 
\begin{enumerate}
    \item $\AGENT$ is a finite set of agents;
    \item $\val_\AGENT$ is  a vector of values (assigning a single value per agent).
   \item  $\probascenario$ is  a probability distribution on 
$2^\AGENT$.
\end{enumerate}
The strategy set $\STRATEGY$ is defined to be the set of probability distributions on $\RR^+$, so a strategy profile $(\strategy_\agent)_{\agent\in\AGENT}\in \STRATEGY^\AGENT$, denoted $\strategy_\AGENT$ induces a probability distribution on $(\RR^+)^\AGENT$.
An $\epsilon$-Nash equilibrium in this setting is a strategy profile
$\strategy_\AGENT^\star$ such that for all $\agent\in\AGENT$ and $\strategy_\AGENT\in \STRATEGY^\AGENT  $
\begin{equation}
    \agentprofit(\strategy^\star_\agent,\strategy_{-\agent}^\star)\geq
       \agentprofit(\strategy_\agent,\strategy_{-\agent}^\star)-\epsilon,
\end{equation}
where, for a first-price auction,
\begin{equation}
\agentprofit(\strategy_\agent,\strategy_{-\agent}) = \int (\value_\agent-\bid_\agent)\prod_{\agent'\in\scenario\setminus \agent}\{\bid_\agent>\bid_{\agent'}\}
\dd\mu(S|\agent\in S)\dd\strategy_{\AGENT}(\bid_\AGENT).
\end{equation}

%% file: equivalence.tex
The relation between the two formulation might not be a priori obvious.
In view of this, we show below that the agent-based formulation is equivalent to the player-based one when the underlying set of values is discrete.

\begin{lemma}
Let $(\PLAYER,I_\PLAYER,F)$ be a player-based one-item auction formulation. Suppose that $I_\player$ is discrete for all $\player\in\PLAYER$.
Then there exist
\begin{enumerate}
    \item An agent-based auction model $(\AGENT,\val_\AGENT,\probascenario)$
    \item A partition $(\AGENT_\player)_{\player\in\PLAYER}$ of $\AGENT$ 
    \item A family of application  $(\phi_\player)_{\player\in\PLAYER}$, where  $\phi_\player$ is a bijection from the player's mixed strategy set   $\Delta\Sigma_\player$ to its agents' strategy set $\Gamma^{\AGENT_\player}$.
\end{enumerate}
such that for any player strategy profile $\beta_\PLAYER\in\prod_{\player\in\PLAYER}\Sigma_\player$, we have 
\begin{equation}
    \pi_{\player}(\beta_\player) =\sum_{\agent\in\AGENT_\player} \pi_\agent((\phi_\player(\beta_\player))_{\player\in\PLAYER})\mu(S\ni\agent)
\end{equation}

\end{lemma}
\begin{proof}
Set $\AGENT_\player=[|I_\player|]$, $v_{\AGENT_\player}= I_\player$ for all $\player\in\PLAYER$, and
$\forall S\in 2^\AGENT,\quad \mu(S)=0\quad \mbox{if}\ S\not\in\times_{\player\in\PLAYER}\AGENT_\player$ and $\mu(S)=F(v_{S})$ otherwise.
Then $\phi_\player$ is obtained by desintegrating the players' mixed strategies over the set of agents.
\end{proof}

%% file: algorithm.tex
As we discussed in the introduced, the proposed heuristic for calculating the Nash equilibrium of an auction is inspired by the process of fictitious play for learning in finite games \citep{Bro51,Rob51}.
The key novelty of our approach --- which leads to an easily implementable heuristic --- is that the process unfolds in the \textit{agent-based} representation of the game, not its normal form (or a discretized version therefof).

To describe the heuristic formally, let $\DISCRETEBIDS$ be a discrete grid of bid values, let $\eta_k>0$ be a "learning rate" sequence, and let $(\gamma^0_{\agent}(\bid))_{\agent\in[\Nagent],\bid \in \DISCRETEBIDS}$ be an initial strategy profile. 
Then the proposed heuristic proceeds by taking at each stage $k=0,1,\dotsc$, a best reply to a generalized exponential moving average of the players' empirical bid distributions up to stage $k$.
More precisely, we have the following sequence of events:
\begin{enumerate}
    \item Strategies are initialized with $(\gamma^0_{\agent}(\bid))_{\agent\in[\Nagent],\bid \in \DISCRETEBIDS}$
    \item At each time step $k$, every agent $\agent\in\AGENT$ performs the bid update
    \begin{equation}
 \bid_\agent^{(k)} \in \argmax_{\bid \in \DISCRETEBIDS}  \agentprofit(\delta_\bid,\strategy_{-\agent}^{(k)})
    \end{equation}
    \item Subsequently, mixed strategies are updated as
    \begin{equation}
        \strategy_{\agent}^{(k+1)} = (1-\eta_k)\strategy_{\agent}^{(k)}+\eta_k\delta_{ \bid_\agent^{(k)}}
    \end{equation}
    \item Repeat until a termination criterion is triggered
\end{enumerate}

At a high level, this process is conceptually similar to fictitious play as, at each stage, the algorithm outputs a best response to "some" average of past play.
This analogy would be precise if the averaging assigned equal weight to all past states, which in turn corresponds to the choice $\eta_k = 1/k$ for the method's learning rate.
However, in many cases, it is more advantageous to assign exponentially more weight to recent observations relative to past ones, as this would effectively discount the irrelevant impact of the method's (arbitrary) initialization;
this corresponds to the choice $\eta_k = \text{const.}$, in which case the averaging undertaken is a straightforward exponential moving average.

%% file: implementation.tex
In terms of implementation, the fictitious bidding heuristic can be coded in less than $100$ lines in Julia~\cite{bezanson2017julia} and is available online.%
\footnote{we will release the code with the definitive version of the paper} All our experiments were run with this implementation on a personal laptop.
For the case of independent probability distributions we can use the simplifying observation  that  $\gamma_{-a}^{(k)}$  is the same for every agent of the same player, which in turn alleviates the computational and memory burden. In the plots we provide, we report the cumulative distribution of the agent's bid as well as their payoff, which allows for checking visually that the support of the bids is approximatelly in the maximand of the payoff.

%% file: examples.tex
\input{example01}
\input{exampleCorr1}
\input{exampleCorr2}
\input{examplewang}
\input{examplewang2}
\input{example_batch}
\input{exemple_second_price}

%% file: example01.tex
\begin{example}
\label{exemple-0-1}
We take
$|\AGENT|=4$, $\val_{\AGENT}=(0,0,1.,1.)$, and
$\mu$ take the value 0.25 on the following 4 scenarios:
$(\agent_1,\agent_2)$,
$(\agent_3,\agent_4)$,
$(\agent_1,\agent_4)$,
$(\agent_2,\agent_3)$.
This corresponds to a symmetric setting with two bidders whose values are sampled in $\{0,1\}$ uniformly and independently.
This example is interesting because it posses an analytic solution.
Indeed, suppose that in the equilibrium, no agent bids above its value, then the bid of agent $\agent_1$ and agent $\agent_2$ should be zero.
Since the setting is symmetric, let us denote by $G$ the cumulative of the bid of agent $\agent_3$ and $\agent_4$.
Observe that \(G(b) = 1/(1-b) - 1 \) for \(b \in[0,1/2]\) corresponds to a Nash equilibrium because (1) the expected payoff of agent $\agent_3$ (resp. $\agent_4$) is constant, equal to 0.5, for \(b \in[0,1/2]\) and strictly smaller  than 0.5  for any \(b >1/2\).
We took $\BB=[0,1/400,2/400,\ldots 1]$, $\eta_k =0.01/k$ and performed $10^5$ iterations.
The strategy estimate we obtain is  $8\times 10^{-5}$-Nash equilibrium.
The result is displayed in Figure~\ref{fig:example0-1}.
\begin{figure}[h]
\includegraphics[width=8cm]{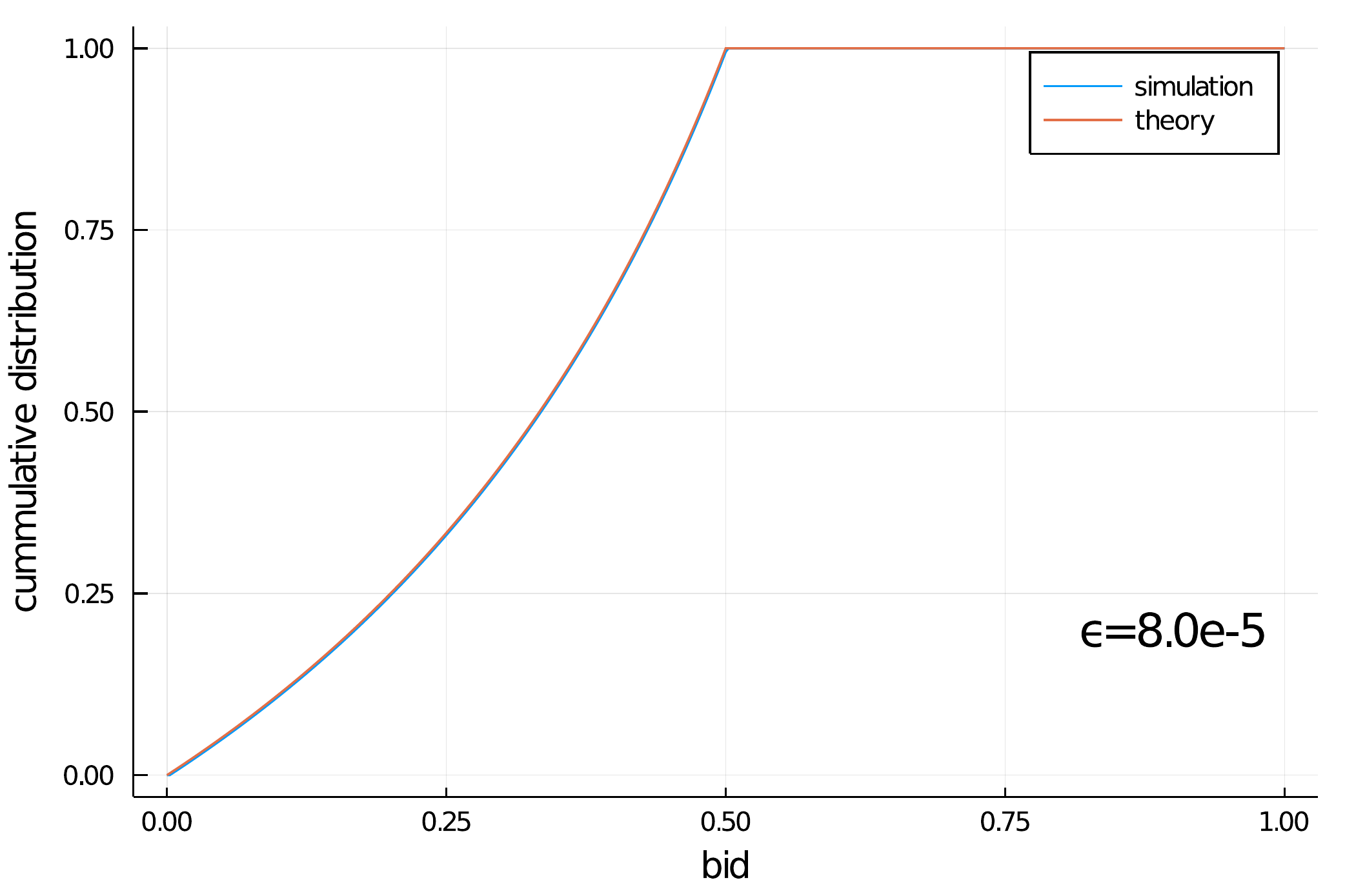}
\caption{The bid cumulative of agent 3 (of value 1)  in Example~\ref{exemple-0-1}}
\label{fig:example0-1}
\end{figure}
\end{example}

%% file: exampleCorr1.tex
\begin{example}
\label{exemple_correlated1}
Next we take $|\AGENT|=3$ and $\value_\AGENT=(1/3,2/3,3/3)$ with the two equiprobable scenarios $(\agent_1, \agent_2)$ and $(\agent_2, \agent_3)$.
We took $\BB=[0,1/600,2/600,\ldots 1]$, $\eta_k =0.01/k$ and performed $10^6$ iterations.
The strategy profile  we obtain is a $(1.5\times 10^{-4})$-Nash equilibrium.
Results are displayed in Figure~\ref{fig:corr1}.
\end{example}
\begin{figure}
     \centering
     \begin{subfigure}[b]{0.3\textwidth}
         \centering
         \includegraphics[width=\textwidth]{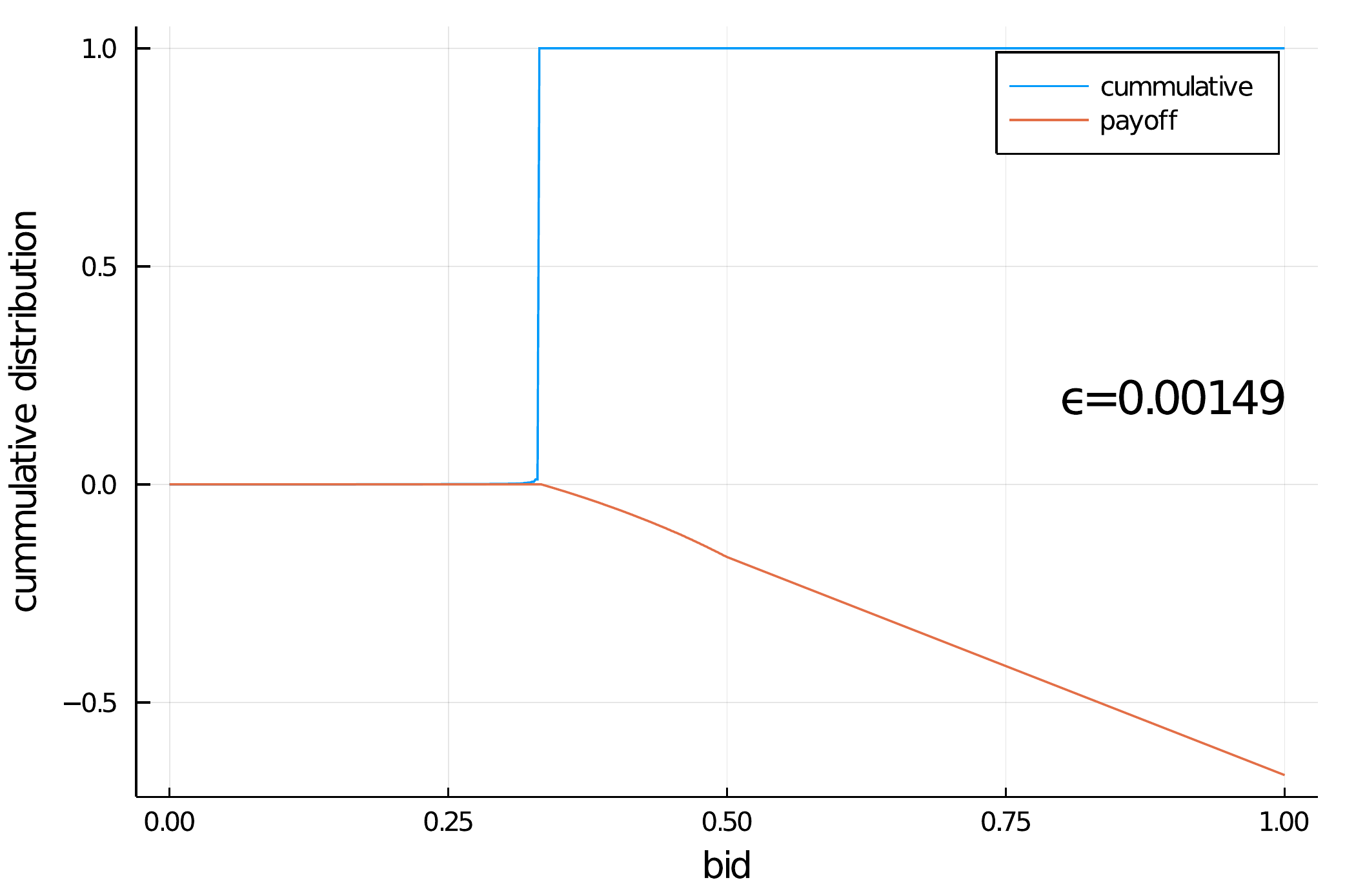}
         \caption{ Agent 1 (of value 1/3)}
         \label{fig:y equals x}
     \end{subfigure}
     \hfill
     \begin{subfigure}[b]{0.3\textwidth}
         \centering
         \includegraphics[width=\textwidth]{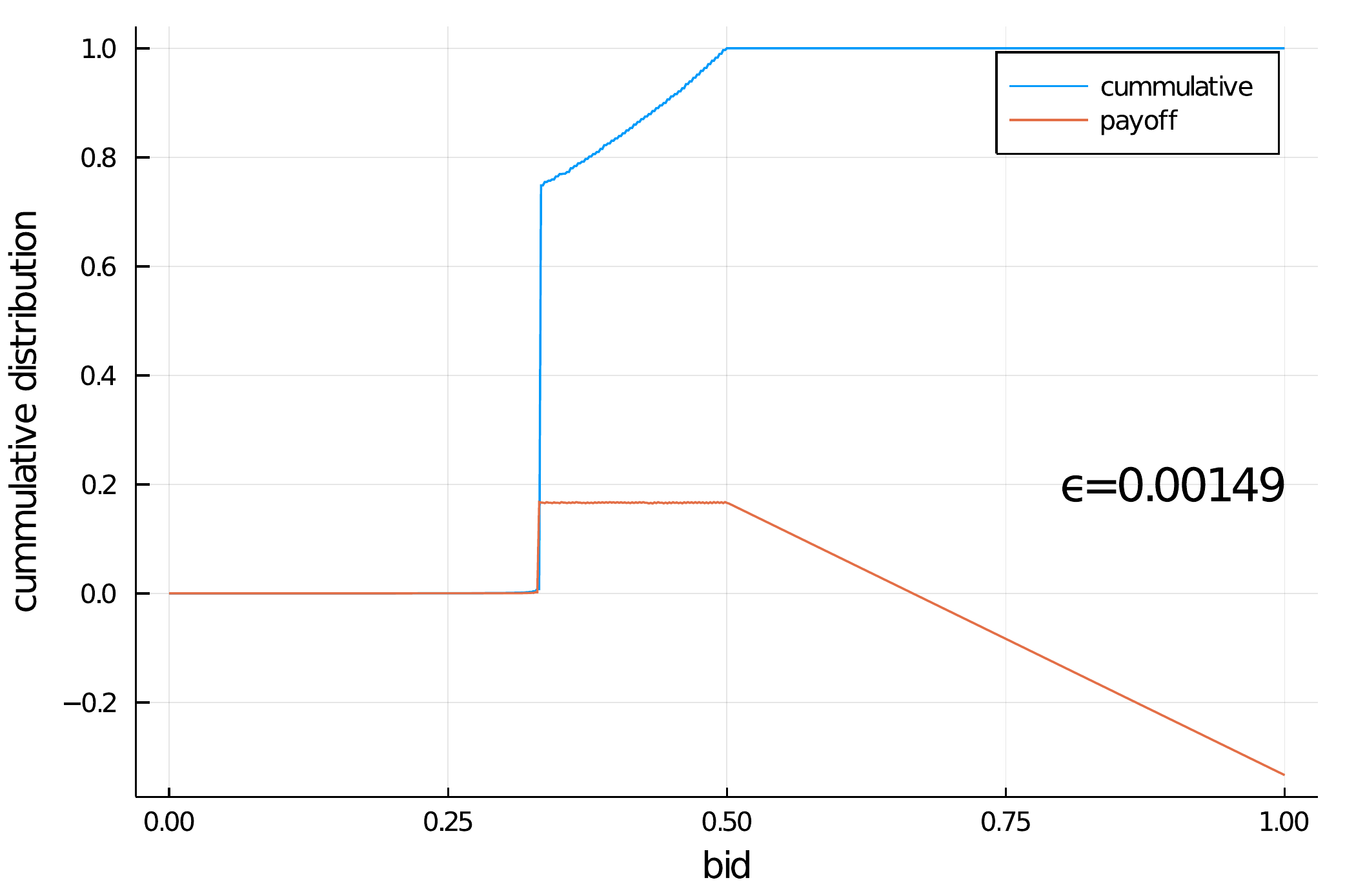}
         \caption{Agent 2 (of value 2/3)}
         \label{fig:three sin x}
     \end{subfigure}
     \hfill
     \begin{subfigure}[b]{0.3\textwidth}
         \centering
         \includegraphics[width=\textwidth]{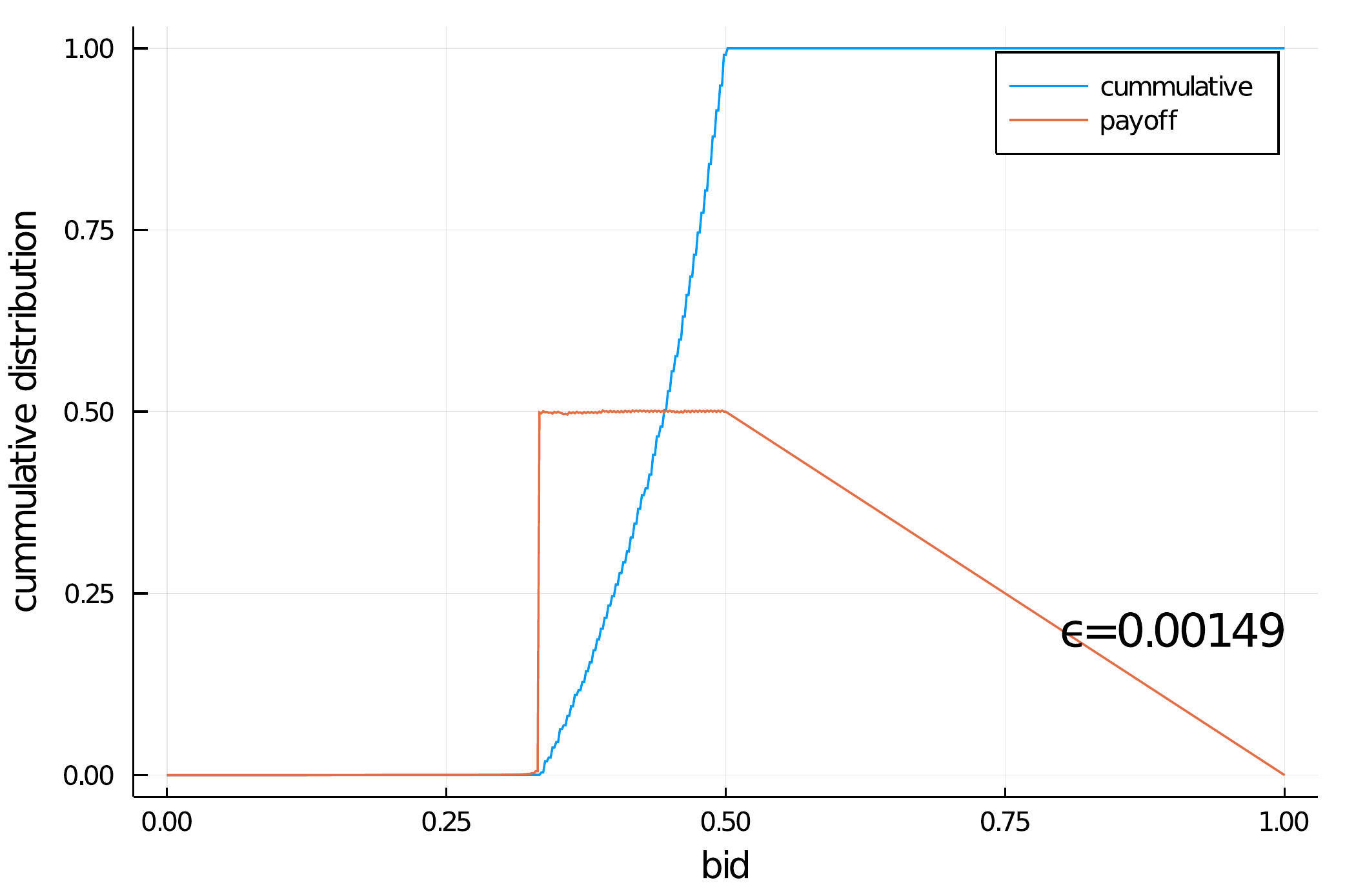}
         \caption{Agent 3 (of value 3/3)}
         \label{fig:five over x}
     \end{subfigure}
        \caption{Results for Example~\ref{exemple_correlated1}}
        \label{fig:corr1}
\end{figure}

%% file: exampleCorr2.tex
\begin{example}
\label{exemple_correlated2}

Let  $\val_\AGENT=(1/4,2/4,2/4,4/4)$ with the four equiprobable scenarii
\((\agent_1,\agent_2),(\agent_2,\agent_3),(\agent_1,\agent_3),(\agent_1,\agent_2,\agent_3,\agent_4) \).
We took $\BB=[0,1/400,2/400,\ldots 1]$, $\eta_k =0.01/k$ and performed $10^5$ iterations.
The strategy estimate we obtain is a $(2.5\times 10^{-3})$-Nash equilibrium.
The results are displayed in Figure~\ref{fig:examplecorr2}
.\end{example}

\begin{figure}
     \centering
     \begin{subfigure}[b]{0.3\textwidth}
         \centering
         \includegraphics[width=\textwidth]{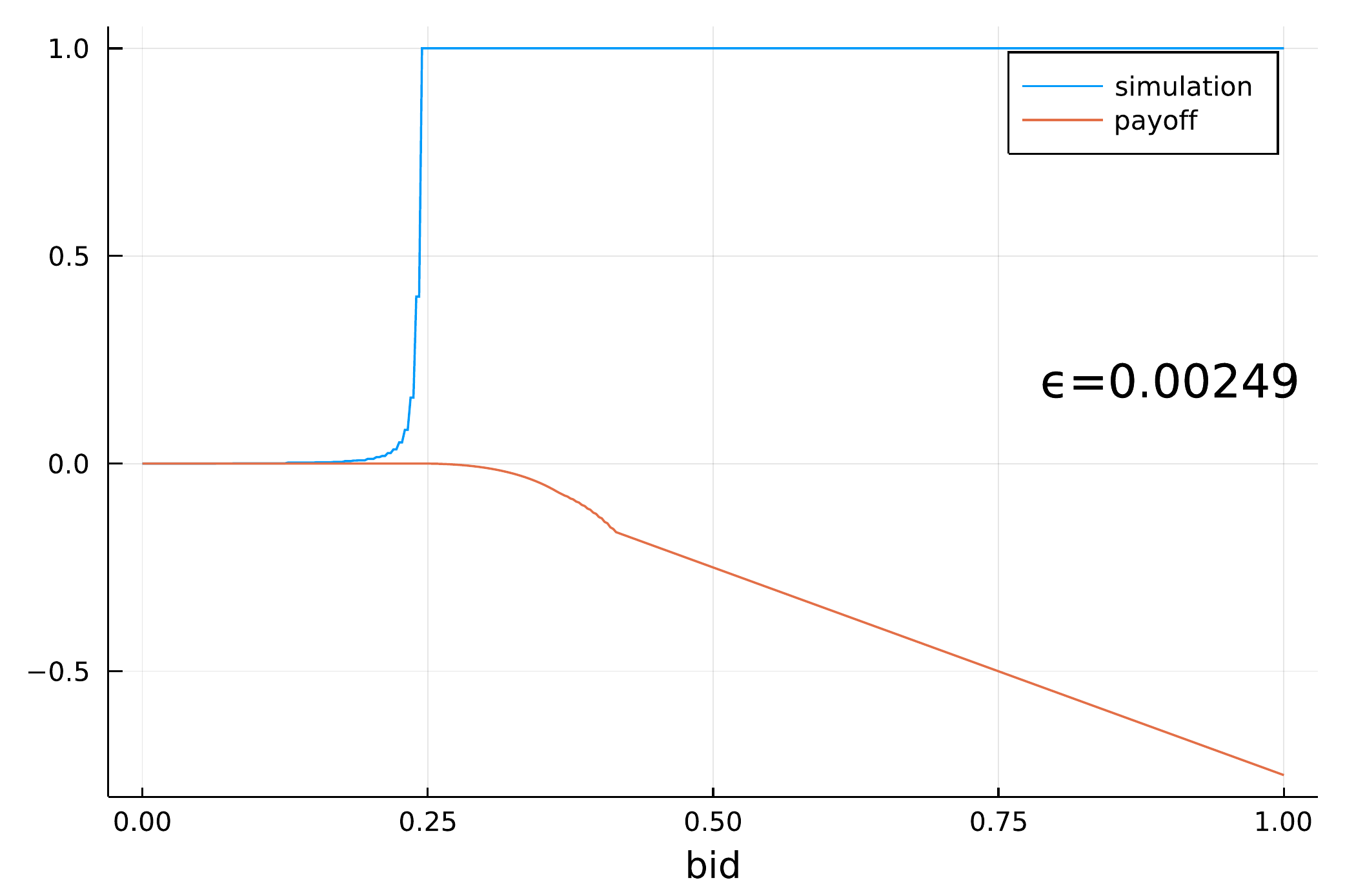}
         \caption{Agent 1 value 1/4}
         \label{fig:y equals x}
     \end{subfigure}
     \hfill
     \begin{subfigure}[b]{0.3\textwidth}
         \centering
         \includegraphics[width=\textwidth]{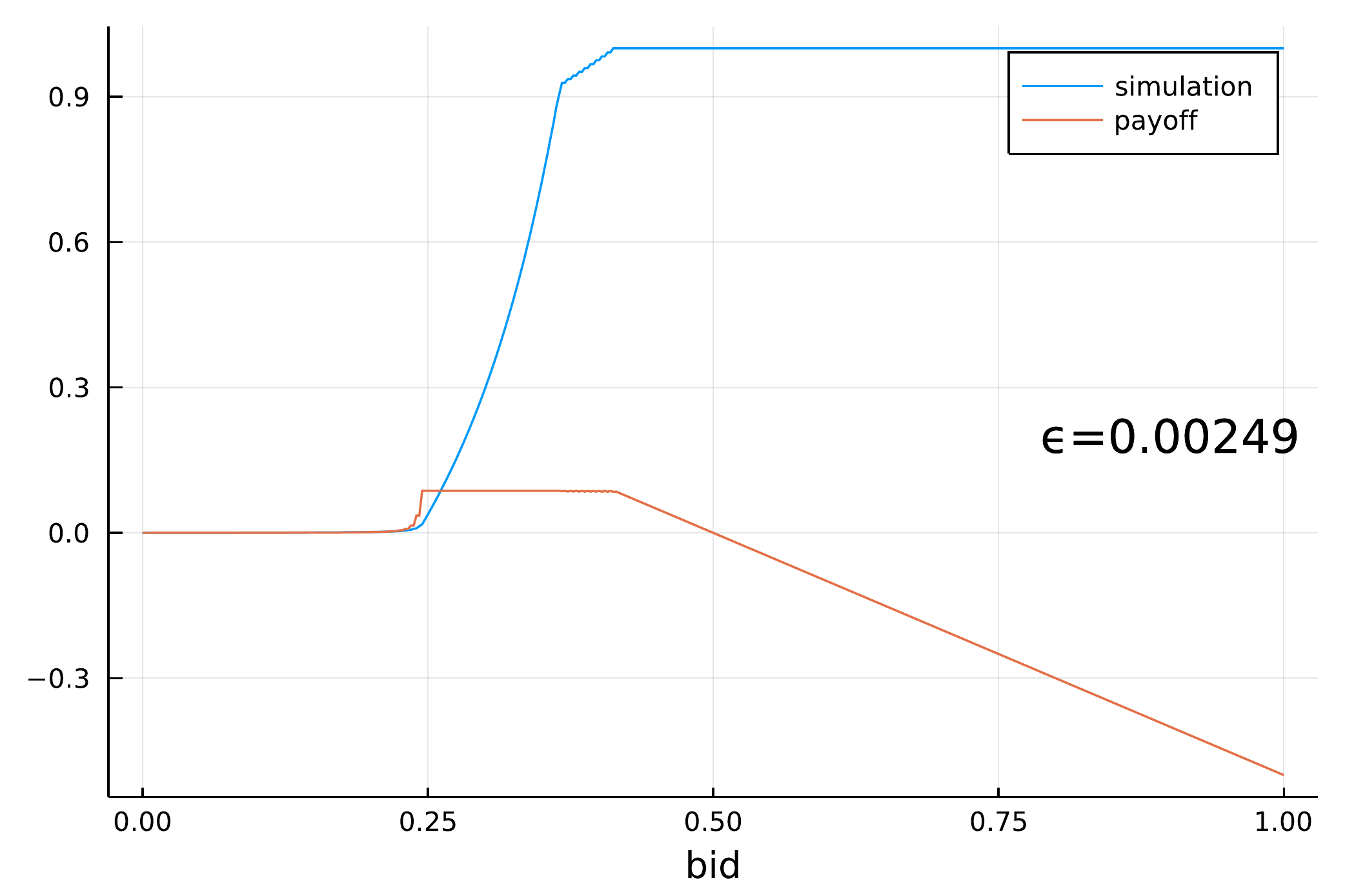}
         \caption{Agent 2 value 2/4}
         \label{fig:three sin x}
     \end{subfigure}
     \hfill
     \begin{subfigure}[b]{0.3\textwidth}
         \centering
         \includegraphics[width=\textwidth]{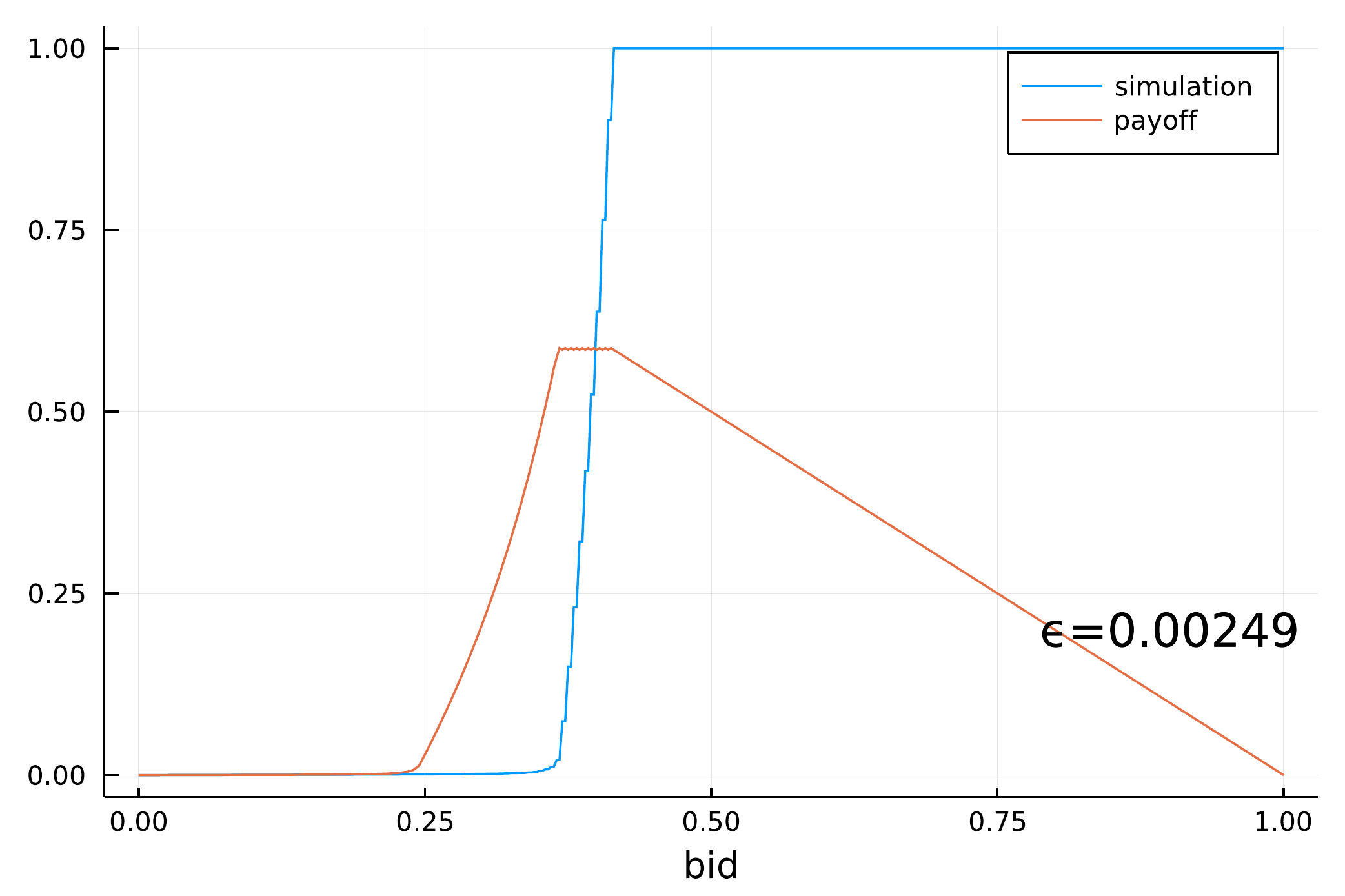}
         \caption{Agent 4  value 4/4}
         \label{fig:five over x}
     \end{subfigure}
        \caption{Example~\ref{exemple_correlated2}}
        \label{fig:examplecorr2}
\end{figure}

%% file: examplewang.tex
\begin{example}[From Wang et al.]
\label{example:wang}
We test the heuristic on Example 8 from~\cite{wang2020bayesian}, which was used to test the stability of continuous method on discrete setting. It is notable that we do  better than the result reported by~\cite{wang2020bayesian} for the continuous methods.
We use the library developed in~\cite{wang2020bayesian} to compare our heuristic with their solution.
The setting is as follow.
There are $3$ players, each represented by $3$ agents, so there are $9$ agents in total.
The potential values of all players belong to $\{0.1,0.2,0.25\}$, and are sampled independently.  What differs between the players is the relative probability of those three values:
\begin{itemize}
    \item For the first player those probabilities are 1/4, 1/4, and 1/2
    \item For the second and third player, they are  0.05, 0.45, and 0.5.
\end{itemize}
The results are reported in Figure~\ref{fig:wang}, and the strategy profile we obtain is a $(4\times 10^{-5})$-Nash equilibrium of the auction.
\end{example}

\begin{figure}
     \centering
     \begin{subfigure}[b]{0.5\textwidth}
         \centering
         \includegraphics[width=\textwidth]{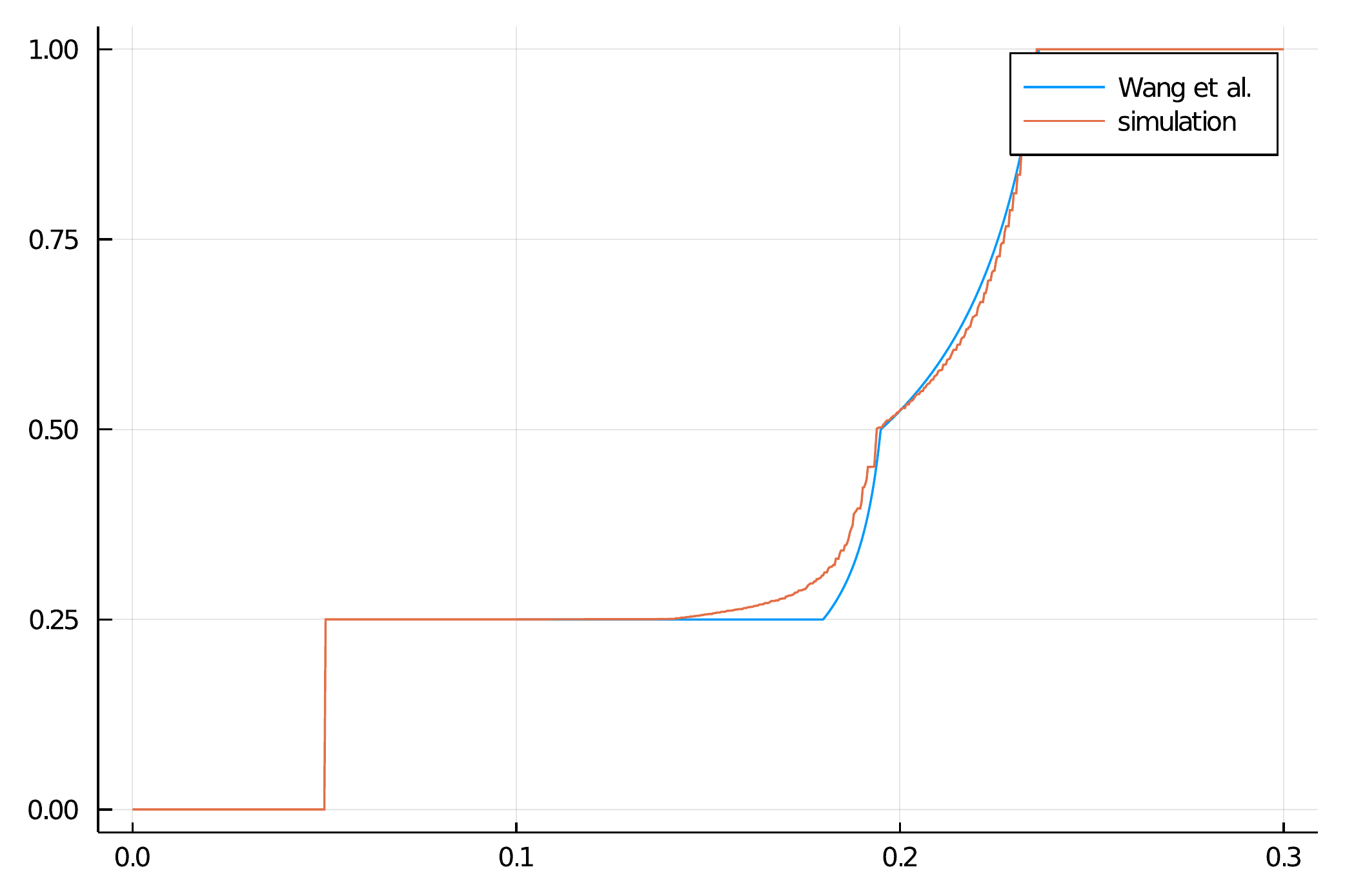}
         \caption{Strategy of player 1}
         \label{fig:y equals x}
     \end{subfigure}
     \hfill
     \begin{subfigure}[b]{0.4\textwidth}
         \centering
         \includegraphics[width=\textwidth]{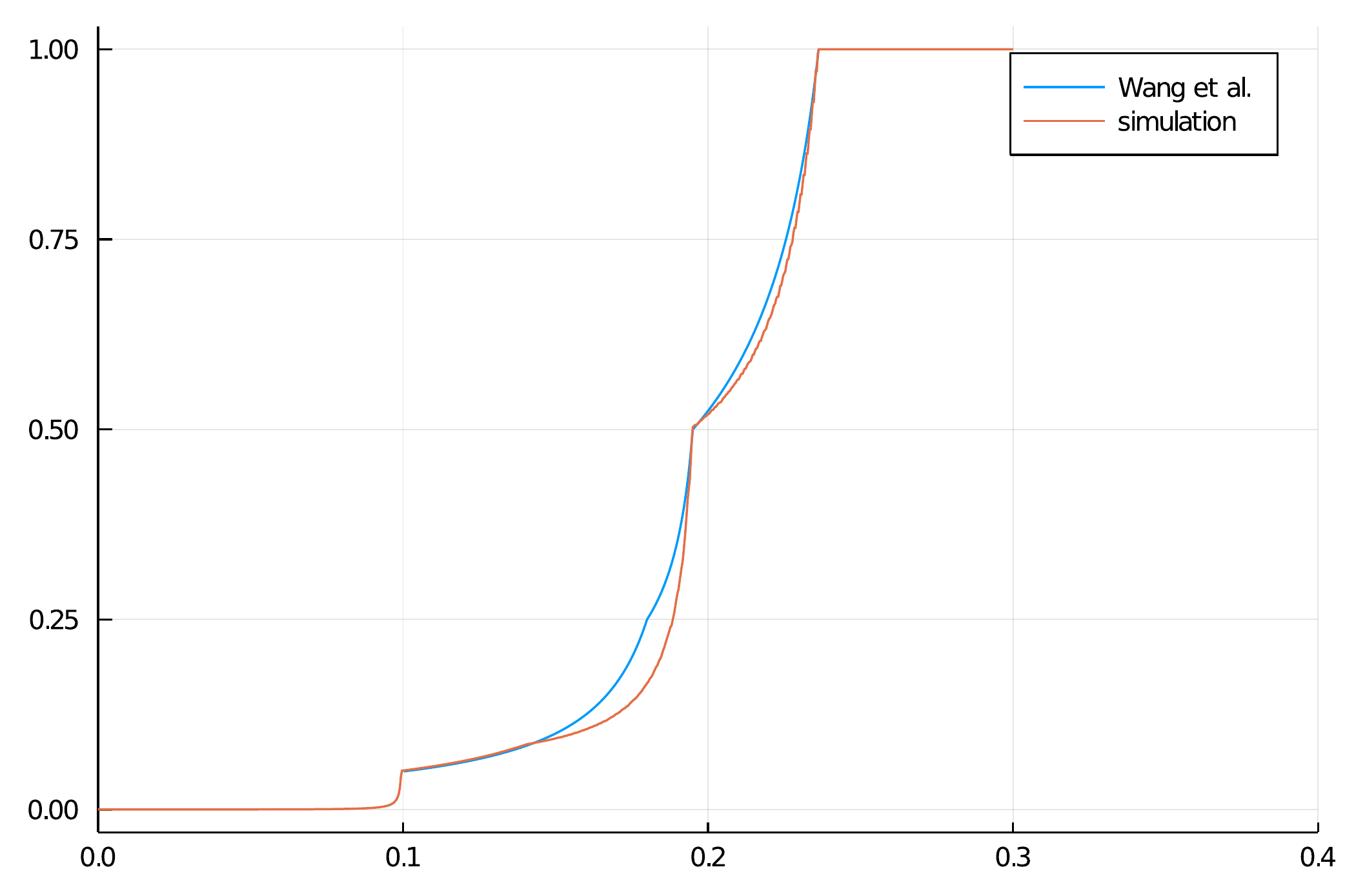}
         \caption{Strategy of player 2 and 3}
         \label{fig:three sin x}
     \end{subfigure}
             \caption{Result for Example~\ref{example:wang}}
        \label{fig:wang}
\end{figure}

%% file: examplewang2.tex
\begin{example}[Another comparison with Wang et al.]
\label{example:wang2}
We use the library developed in~\cite{wang2020bayesian} to compare fictitious bidding with another example.
The setting is as follow.
There are 2 players, 2 agents for player 1, 3 for player 2. .
The potential values of player 1 are
 0.1 and 0.25 with respective probability 0.25 and 0.75, and the potential values of player 2 are 0.1, 0.2 and 0.25 with respective probability 0.05, 0.45 and 0.5.
The values are sampled independently.
The results are reported in Figure~\ref{fig:wang2} and the strategy profile we obtain is a $(9\times 10^{-4})$-Nash equilibrium.
\end{example}

\begin{figure}
     \centering
     \begin{subfigure}[b]{0.4\textwidth}
         \centering
         \includegraphics[width=\textwidth]{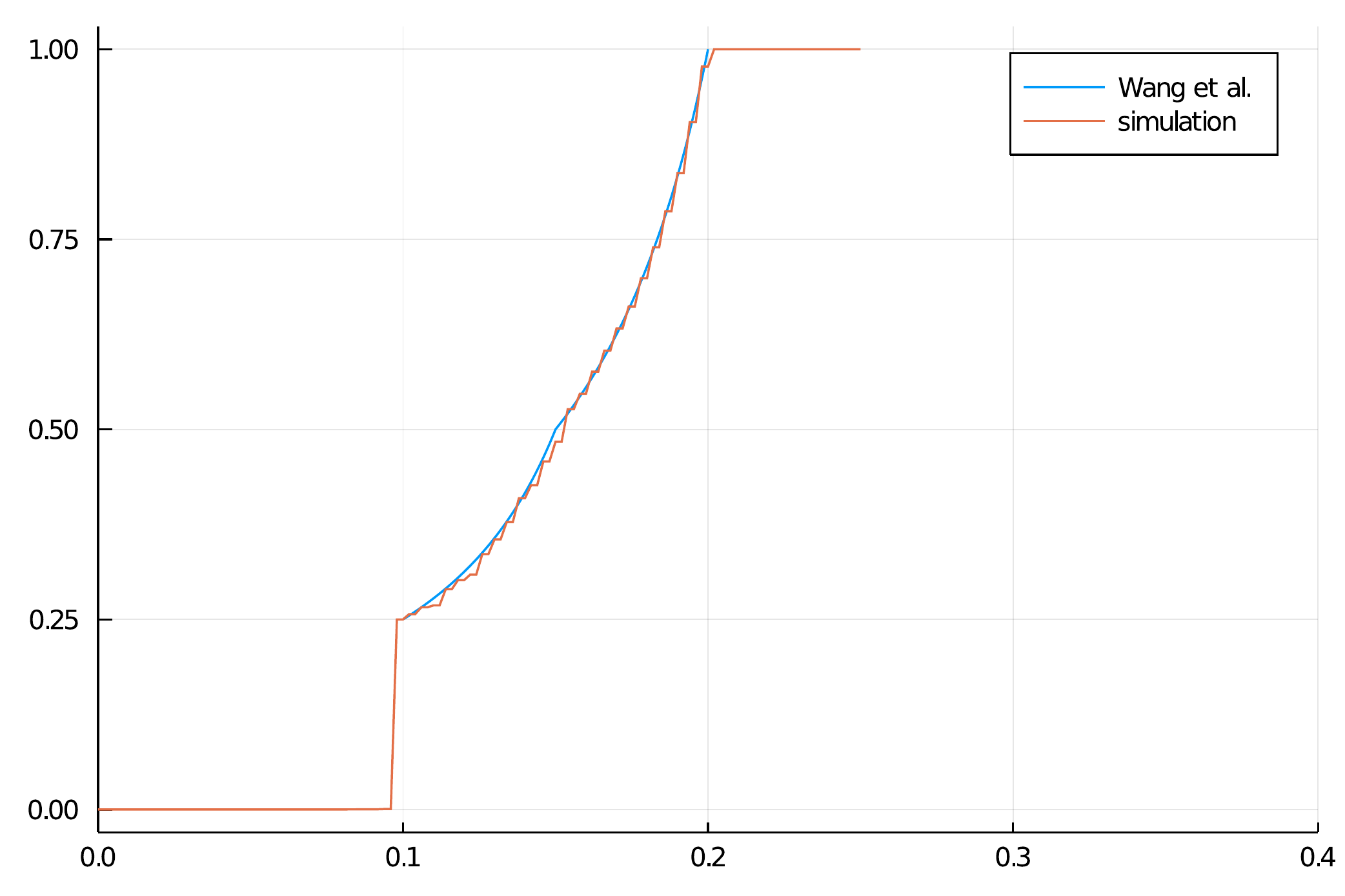}
         \caption{Strategy of player 1}
         \label{fig:y equals x}
     \end{subfigure}
     \hfill
     \begin{subfigure}[b]{0.4\textwidth}
         \centering
         \includegraphics[width=\textwidth]{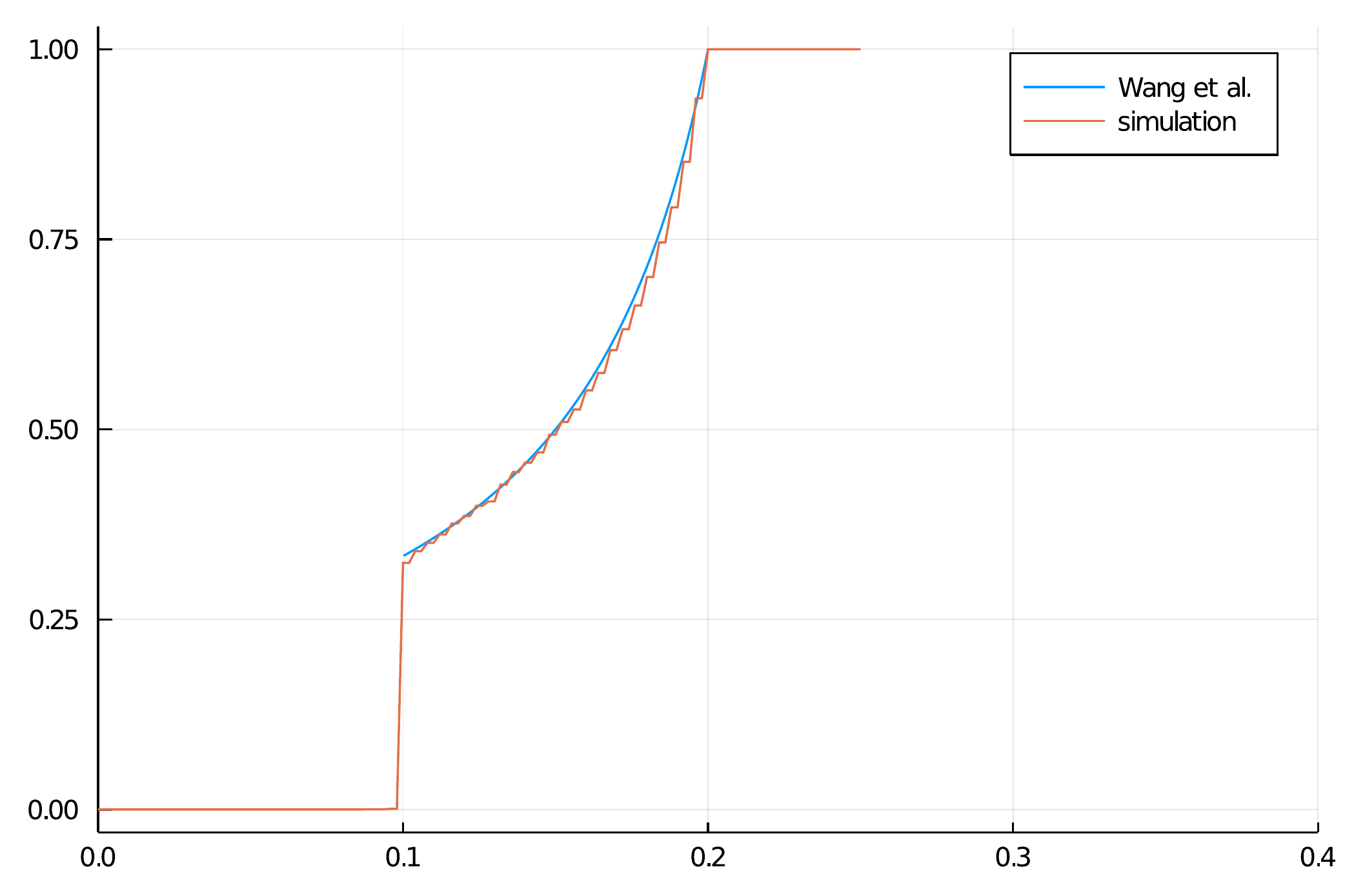}
         \caption{Strategy of player 2}
         \label{fig:three sin x}
     \end{subfigure}
             \caption{Result for Example~\ref{example:wang2}}
        \label{fig:wang2}
\end{figure}

%% file: example_batch.tex
\begin{example}
\label{exemple:batch}
In order to test the robustness of the fictitious bidding, we generated random instances of auction games. 
Each instance contains 10 agents whose values are samples uniformly between 0 and 1. 
We then generate 20 scenarios by taking ten 20 times at random a pair in $\{1\ldots 10\}$.
We took $\BB=[0,1/100,2/100,\ldots 1]$, $\eta_k =0.01/k$ and performed $10^6$ iterations.
The strategy profile we obtain is  $\epsilon$-Nash equilibrium, with $\epsilon$ respectively equals to  0.00048, 0.00083, 0.00139, 0.00058, 0.00402, 0.01374, 0.00501, 0.00305, 0.00326 and 0.00472.
\end{example}

%% file: exemple_second_price.tex
As a side note, we observed during our experiments that the fictitious bidding heuristic seems to work well for other auctions, such as mixture of first and second-price (The item is still allocated to the highest bidder, but the winner now pays the average of the first and second highest bids).

%% file: discussion.tex
We have introduced a fictitious-play like heuristic to compute the Bayes-Nash equilibrium of first-price auctions even when the values are correlated: fictitious bidding.
Compared to other methods in the literature, the proposed fictitious bidding is inherently discrete, as it is introduced for auctions with values and bids living in a discrete set.
As a result, our heuristic
does not rely on the traditional continuous approximation of the auction, which can be viewed in itself as an improvement, since applications usually take place in discrete environments. 
The heuristic displayed also good result for mixtures of first and second-price auctions.
Overall, our experimental results suggest there is a strong convergence theorem underneath;
clarifying under which conditions this is true (and why) is a central question for future research on the topic.

%% file: main.bbl
\begin{thebibliography}{10}
\expandafter\ifx\csname url\endcsname\relax
  \def\url#1{\texttt{#1}}\fi
\expandafter\ifx\csname urlprefix\endcsname\relax\def\urlprefix{URL }\fi
\expandafter\ifx\csname href\endcsname\relax
  \def\href#1#2{#2} \def\path#1{#1}\fi

\bibitem{CP14}
Y.~Cai, C.~H. Papadimitriou, Simultaneous {Bayesian} auctions and computational
  complexity, in: EC' 14: Proceedings of the 15th ACM Conference on Electronic
  Commerce, 2014.

\bibitem{krishna2009auction}
V.~Krishna, Auction theory, Academic press, 2009.

\bibitem{vickrey1961counterspeculation}
W.~Vickrey, Counterspeculation, auctions, and competitive sealed tenders, The
  Journal of finance 16~(1) (1961) 8--37.

\bibitem{plum1992characterization}
M.~Plum, Characterization and computation of nash-equilibria for auctions with
  incomplete information, International Journal of Game Theory 20~(4) (1992)
  393--418.

\bibitem{marshall1994numerical}
R.~C. Marshall, M.~J. Meurer, J.-F. Richard, W.~Stromquist, Numerical analysis
  of asymmetric first price auctions, Games and Economic Behavior 7~(2) (1994)
  193--220.

\bibitem{10.2307/2648842}
B.~Lebrun, \href{http://www.jstor.org/stable/2648842}{First price auctions in
  the asymmetric n bidder case}, International Economic Review 40~(1) (1999)
  125--142.
\newline\urlprefix\url{http://www.jstor.org/stable/2648842}

\bibitem{maskin2003uniqueness}
E.~Maskin, J.~Riley, Uniqueness of equilibrium in sealed high-bid auctions,
  Games and Economic Behavior 45~(2) (2003) 395--409.

\bibitem{reny2004existence}
P.~J. Reny, S.~Zamir, On the existence of pure strategy monotone equilibria in
  asymmetric first-price auctions, Econometrica 72~(4) (2004) 1105--1125.

\bibitem{bajari2001comparing}
P.~Bajari, Comparing competition and collusion: a numerical approach, Economic
  Theory 18~(1) (2001) 187--205.

\bibitem{Gayle2008}
W.~R. Gayle, J.~F. Richard, Numerical solutions of asymmetric, first-price,
  independent private values auctions, Computational Economics 32 (2008)
  245--278.
\newblock \href {http://dx.doi.org/10.1007/s10614-008-9125-7}
  {\path{doi:10.1007/s10614-008-9125-7}}.

\bibitem{fibich2011numerical}
G.~Fibich, N.~Gavish, Numerical simulations of asymmetric first-price auctions,
  Games and Economic Behavior 73~(2) (2011) 479--495.

\bibitem{kaplan2012asymmetric}
T.~R. Kaplan, S.~Zamir, Asymmetric first-price auctions with uniform
  distributions: analytic solutions to the general case, Economic Theory 50~(2)
  (2012) 269--302.

\bibitem{Fibich2012}
G.~Fibich, N.~Gavish, \href{https://doi.org/10.1287/moor.1110.0535}{Asymmetric
  first-price auctions—a dynamical-systems approach}, Mathematics of
  Operations Research 37~(2) (2012) 219--243.
\newblock \href {http://arxiv.org/abs/https://doi.org/10.1287/moor.1110.0535}
  {\path{arXiv:https://doi.org/10.1287/moor.1110.0535}}, \href
  {http://dx.doi.org/10.1287/moor.1110.0535}
  {\path{doi:10.1287/moor.1110.0535}}.
\newline\urlprefix\url{https://doi.org/10.1287/moor.1110.0535}

\bibitem{hubbard2014numerical}
T.~P. Hubbard, H.~J. Paarsch, On the numerical solution of equilibria in
  auction models with asymmetries within the private-values paradigm, in:
  Handbook of computational economics, Vol.~3, Elsevier, 2014, pp. 37--115.

\bibitem{10.5555/3398761.3398929}
Z.~Wang, W.~Shen, S.~Zuo, Bayesian nash equilibrium in first-price auction with
  discrete value distributions, in: Proceedings of the 19th International
  Conference on Autonomous Agents and MultiAgent Systems, AAMAS '20,
  International Foundation for Autonomous Agents and Multiagent Systems,
  Richland, SC, 2020, p. 1458–1466.

\bibitem{heymann2020bid}
B.~Heymann, How to bid in unified second-price auctions when requests are
  duplicated, Operations Research Letters 48~(4) (2020) 446--451.

\bibitem{paes2020competitive}
R.~Paes~Leme, B.~Sivan, Y.~Teng, Why do competitive markets converge to
  first-price auctions?, in: Proceedings of The Web Conference 2020, 2020, pp.
  596--605.

\bibitem{filos2021complexity}
A.~Filos-Ratsikas, Y.~Giannakopoulos, A.~Hollender, P.~Lazos, D.~Po{\c{c}}as,
  On the complexity of equilibrium computation in first-price auctions, arXiv
  preprint arXiv:2103.03238.

\bibitem{rasooly2021importance}
I.~Rasooly, C.~Gavidia-Calderon, The importance of being discrete: on the
  inaccuracy of continuous approximations in auction theory (2021).
\newblock \href {http://arxiv.org/abs/2006.03016} {\path{arXiv:2006.03016}}.

\bibitem{bichler2021learning}
M.~Bichler, M.~Fichtl, S.~Heiderkrueger, N.~Kohring, P.~Sutterer, Learning
  equilibria in symmetric auction games using artificial neural networks,
  Nature Machine Intelligence.

\bibitem{myerson1981optimal}
R.~B. Myerson, Optimal auction design, Mathematics of operations research 6~(1)
  (1981) 58--73.

\bibitem{hartline2009simple}
J.~D. Hartline, T.~Roughgarden, Simple versus optimal mechanisms, in:
  Proceedings of the 10th ACM conference on Electronic commerce, 2009, pp.
  225--234.

\bibitem{dhangwatnotai2015revenue}
P.~Dhangwatnotai, T.~Roughgarden, Q.~Yan, Revenue maximization with a single
  sample, Games and Economic Behavior 91 (2015) 318--333.

\bibitem{cole2014sample}
R.~Cole, T.~Roughgarden, The sample complexity of revenue maximization, in:
  Proceedings of the forty-sixth annual ACM symposium on Theory of computing,
  2014, pp. 243--252.

\bibitem{Bro51}
G.~W. Brown, Iterative solutions of games by fictitious play, in: T.~C.
  Coopmans (Ed.), Activity Analysis of Productions and Allocation, 374-376,
  Wiley, 1951.

\bibitem{Rob51}
J.~Robinson, An iterative method for solving a game, Annals of Mathematics 54
  (1951) 296--301.

\bibitem{Miy61}
K.~Miyasawa, On the convergence of learning processes in a $2 \times 2$
  non-zero sum game, Research memorandum~33, Princeton University (1961).

\bibitem{MS96-jet}
D.~Monderer, L.~S. Shapley, Fictitious play property for games with identical
  interests, Journal of Economic Theory 68 (1996) 256--265.

\bibitem{Hof95b}
J.~Hofbauer, Stability for the best response dynamics, mimeo (1995).

\bibitem{MR90}
P.~Milgrom, J.~Roberts, Rationalizability, learning, and equilibrium in games
  with strategic complementarities, Econometrica 58~(6) (1990) 1255--1277.

\bibitem{MR91}
P.~Milgrom, J.~Roberts, Adaptive and sophisticated learning in normal form
  games, Games and Economic Behavior 3 (1991) 82--100.

\bibitem{Kri92}
V.~Krishna, Learning in games with strategic complementarities, mimeo (1992).

\bibitem{Hah08}
S.~Hahn, The convergence of fictitious play in games with strategic
  complementarities, Economics Letters 99~(304-306).

\bibitem{FK93}
D.~Fudenberg, D.~M. Kreps, Learning mixed equilibria, Games and Economic
  Behavior 5~(320-367).

\bibitem{HS09}
J.~Hofbauer, W.~H. Sandholm, Stable games and their dynamics, Journal of
  Economic Theory 144~(4) (2009) 1665--1693.

\bibitem{SS11}
S.~Shalev-Shwartz, Online learning and online convex optimization, Foundations
  and Trends in Machine Learning 4~(2) (2011) 107--194.

\bibitem{MS16}
P.~Mertikopoulos, W.~H. Sandholm, Learning in games via reinforcement and
  regularization, Mathematics of Operations Research 41~(4) (2016) 1297--1324.

\bibitem{HLMS21a}
S.~Hadikhanloo, R.~Laraki, P.~Mertikopoulos, S.~Sorin, Learning in nonatomic
  games, {Part I}: {Finite} action spaces and population games,
  \url{https://arxiv.org/abs/2107.01595} (2021).

\bibitem{Jor93}
J.~S. Jordan, Three problems in learning mixed strategy {Nash} equilibria,
  Games and Economic Behavior 5~(3) (1993) 368--386.

\bibitem{Sha64}
L.~S. Shapley, Some topics in two-person games, in: Advances in Game Theory,
  no.~52 in Annals of Mathematics Studies, Princeton University Press, 1964.

\bibitem{GH95}
A.~Gaunersdorfer, J.~Hofbauer, Fictitious play, {Shapley} polygons, and the
  replicator equation, Games and Economic Behavior 11~(2) (1995) 279--303.

\bibitem{lebrun1999first}
B.~Lebrun, First price auctions in the asymmetric n bidder case, International
  Economic Review 40~(1) (1999) 125--142.

\bibitem{athey2001single}
S.~Athey, Single crossing properties and the existence of pure strategy
  equilibria in games of incomplete information, Econometrica 69~(4) (2001)
  861--889.

\bibitem{wang2020bayesian}
Z.~Wang, W.~Shen, S.~Zuo, Bayesian nash equilibrium in first-price auction with
  discrete value distributions, in: Proceedings of the 19th International
  Conference on Autonomous Agents and MultiAgent Systems, 2020, pp. 1458--1466.

\bibitem{aummanmixed}
R.~Aumman, Mixed and behavior strategies in infinite extensive games, advances
  in game theory, Annals of Mathematics Studies  627--650.

\bibitem{bezanson2017julia}
J.~Bezanson, A.~Edelman, S.~Karpinski, V.~B. Shah,
  \href{https://doi.org/10.1137/141000671}{Julia: A fresh approach to numerical
  computing}, SIAM review 59~(1) (2017) 65--98.
\newline\urlprefix\url{https://doi.org/10.1137/141000671}

\end{thebibliography}
